\documentclass[3p, sort&compress]{elsarticle}
\usepackage{amsmath}
\usepackage{amsfonts}
\usepackage{amssymb}
\usepackage{amsthm}
\usepackage[utf8]{inputenc}
\usepackage{cellspace}
\usepackage[small,sc]{caption}
\usepackage{booktabs, multicol, multirow}
\usepackage{enumitem}

\usepackage{graphicx}
\usepackage{epsfig}
\usepackage{algorithm}
\usepackage{algpseudocode}
\usepackage{color}

\usepackage{lineno}
\usepackage{extpfeil}

\usepackage{subdepth}

\theoremstyle{plain}
\newtheorem{theorem}{Theorem}

\newtheorem{proposition}{Proposition}

\theoremstyle{definition}
\newtheorem{definition}{Definition}

\newtheorem{example}{Example}

\theoremstyle{definition}
\newtheorem{remark}{Remark}

\usepackage{geometry}                		
\usepackage{graphicx}				
\usepackage{amssymb}

\journal{MSCS}

\begin{document}
\begin{frontmatter}
\title{ \Large{Learning Quantum Finite Automata with Queries}}

\author{Daowen Qiu\corref{one}}

\cortext[one]{ issqdw@mail.sysu.edu.cn (D. Qiu)}

\address{Institute of Quantum Computing and Computer Theory,
School of
Computer Science and Engineering,\\ Sun Yat-sen University, Guangzhou 510006, China}

\begin{abstract}
{\it Learning finite automata} (termed as {\it model learning}) has become an important field in machine learning and has been useful realistic applications.
Quantum finite automata (QFA) are simple models of quantum computers with finite memory. Due to their simplicity, QFA have well physical realizability, but one-way
QFA still have essential advantages over classical finite automata with regard to state complexity (two-way QFA are more powerful than classical finite automata in computation ability as well). As a different problem in {\it quantum learning theory} and {\it quantum machine learning},  in this paper, our purpose is to initiate the study of {\it learning QFA with queries} (naturally it may be termed as {\it quantum model learning}), and the main results are  regarding learning two basic one-way QFA (1QFA): (1) We propose a  learning algorithm for measure-once 1QFA (MO-1QFA) with query complexity of polynomial time; (2) We propose a  learning algorithm for measure-many 1QFA (MM-1QFA) with query complexity of polynomial-time, as well.  



\end{abstract}

\end{frontmatter}

\section{Introduction}

{\it Learning finite automata} has become an important field in machine learning \cite {KV94} and has been applied to wide-ranging realistic problems \cite{Hig05, Vaa17},  for example, smartcards, network protocols, legacy software,   robotics and control systems, pattern recognition, computational linguistics, computational biology, data compression, data mining, etc. In \cite{Vaa17} learning finite automata is termed as {\it model learning}. In fact, model learning and model checking as well as and model-based testing have intrinsic connections (see the pioneering contribution of Peled et al. \cite{Pel02} and Steffen et al. \cite{Hig10}).


Learning finite automata was first considered by Moore \cite {Moore56} and an exponential-time query algorithm was proposed.
In particular,  Angluin \cite{Ang87} in 1987 proposed the so-called membership and equivalence queries,  a ground-breaking method for learning the models of finite automata. For learning {\it deterministic finite automata} (DFA),  according to Angluin's algorithm,  the learner initially only knows the inputs (i.e. alphabet)  of the model to be learned (say ${\cal M}$), and the aim of the learner
is to learn the model by means of two types of queries, that is, membership queries (MQ) and equivalence queries (EQ). MQ means that  the learner asks what the result  (accepting or rejecting) of output is
in response to an input sequence, and the oracle answers with accepting or rejecting, while EQ signifies the learner  whether a hypothesized machine model (say ${\cal H}$) is the same as the learned machine, and the oracle answers yes if this is the case. Otherwise `no" is replied  and  an input string is provided as a counterexample to distinguish ${\cal H}$ and ${\cal M}$.

 The complexity of queries of Angluin's algorithm \cite{Ang87} is polynomial for learning  DFA and Mealy machines. Angluin \cite{Ang88}  proved that DFA cannot be learned in polynomial time by  membership queries (or equivalence queries) only. Since Angluin's algorithm was proposed \cite{Ang87}, learning other models of finite automata has been investigated. Tzeng \cite{Tze92} studied learning {\it probabilistic finite automata} (PFA) and Markov chains via {\it SD oracle},  where SD oracle can answer state distribution, i.e., probability distribution of states for each input string, so it is more powerful than  membership queries (MQ).  For learning DFA via SD oracle, a state is replied for each input string, and the query complexity of learning DFA via  SD oracle is polynomial \cite{Tze92}.

 Then Bergadano and  Varricchio \cite{BV96} used membership queries (MQ) and equivalence queries (EQ) to learn appropriately probabilistic finite automata, and  a {\it probably approximately correctly} learning algorithm (i.e. PAC algorithm) was presented. Learning {\it nondeterministic finite automata} (NFA) was studied by
 Bollig et al. \cite{BHK09}. In recent years, Angluin et al. \cite{AEF15} initiated the research of learning alternating automata, and Berndt et al. \cite{BLL17} further solved the learning problem of  residual alternating automata.

A natural inquiry is that SD oracle seems too strong.
 However,  it was showed by Tzeng \cite{Tze92} that SD oracle is actually not too strong for learning DFA and PFA if the query complexity is required to be polynomial, because learning a consistency problem related to DFA and PFA  via  SD oracle is still NP-complete \cite{Tze92}. In this paper, we use an {\it AD oracle}  for learning {\it quantum finite automata} (QFA) in polynomial time, that is, AD oracle can answer a state of superposition for each input string, i.e., amplitude distribution of states. 
 Similarly it follows that using {\it AD oracle}  to learn a consistency problem related to reversible finite automata (RFA) and MO-1QFA   is NP-complete.




{\it Quantum machine learning} was early considered by  Bshouty and Jackson \cite{BJ99}  with learning from quantum examples, and  then {\it quantum learning theory} \cite{AW17}  as an important theoretical subject of quantum machine learning has been deeply developed. Quantum learning theory \cite{AW17}  includes models of quantum exact learning, quantum PAC learning, and quantum agnostic learning;  these models are combinations of corresponding classical learning models with quantum computing (in a way, quantum query algorithms). We further review quantum learning theory and quantum machine learning more specifically.

We first recall quantum learning theory, which studies the theoretical aspects of quantum machine learning. 
As pointed out above, in 1999, Bshouty and Jackson \cite{BJ99}  showed that all probably approximately correct (PAC)-learnable function classes are learnable in the quantum models, and notably,
in 2004, Servedio and Gortler \cite{SG04} studied quantum versions of  Angluin’s model of exact learning from membership queries and Valiant’s PAC model of learning from random examples.  Then,
in 2007,  Aaronson \cite{Aa07} investigated learning quantum states, and in 2010,  Zhang \cite{Zha10} further investigated the quantum PAC learning model.  
In 2012,  Gavinsky \cite{Ga12} initialed a new quantum learning model called Predictive Quantum (PQ), which is the quantum analogue of PAC, and afterwards, in 2015,  Belovs \cite{Be15} investigated  the junta learning problem by designing quantum algorithms. 
Quantum deep learning was studied by  Wiebe et al. \cite{WKS16} in 2016, and 
Cheng et al. \cite{CHY16}   provided a  framework to analyze learning matrices in the Schatten class. A detailed survey  concerning quantum learning theory was  presented by
Arunachalam and de Wolf \cite{AW17} in 2017,  and they \cite{AW18} further showed that classical and quantum sample complexity are equal up to constant factors for every concept class.

Now we  simply recall the development of quantum machine learning (QML).
In 2009, Harrow et al. \cite{HH09} proposed a quantum algorithm for solving systems of linear equations, which may be thought of the start of studying quantum machine learning. Then, in 2012, Wiebe et al. \cite{WD12} proposed a quantum linear regression algorithm by virtue of  HHL algorithm. In 2014, Lloyd et al. \cite{LR14}  proposed a quantum version of PCA (principal component analysis) dimension reduction algorithm. Also, quantum matrix inversion was employed in a supervised discriminative learning algorithm \cite{Rebentrost 2014}.  In 2015, Schuld et al. \cite{Schuld  2015} presented a comprehensive perspective on quantum machine learning, and in 2016, Cong et al. \cite{CD16} proposed a quantum data dimension reduction algorithm. In 2017, Biamonte et al.  \cite{Biamonte  2017} focused on utilizing a quantum computer to analyze classical or quantum data encoded as quantum states, and Kerenidis and Prakash\cite{KP17} proposed a quantum algorithm for recommendation systems.  In 2018, Lloyd et al. \cite{LW18} changed a classical generative adversarial network to obtain a quantum generative adversarial network, and Mitarai et al. \cite{MM18} constructed a quantum neural network model.  In 2019, Zhao et al. \cite{Zh19} designed quantum Bayesian neural networks, and Benedetti et al. \cite{Benedetti 2019} provided an overview of these models' components and investigated their application.  Recently, in 2021, an application of machine learning techniques to quantum devices was found by Marquardt et al. \cite{Marquardt 2021}.   In 2022, Huang et al. \cite{HB22} confirmed that QML can more effectively learn the operating rules of the physical world than any classical machine learning method, and then, prediction of the evolution of quantum systems has been achieved successfully in \cite{Rodriguez 2022}. More recently, in 2023, Meyer et al. \cite{Meyer 2023} explored how symmetries in learning problems can be exploited to create quantum learning models, and Krenn et al. \cite{Krenn 2023} discussed the application of machine learning and artificial intelligence in analyzing quantum measurements.

However, {\it learning quantum finite automata (QFA)}  is still a pending problem to be studied, and this is the main goal of this paper. QFA can be thought of as a theoretical model of quantum computers in which the memory is finite and described by a finite-dimensional state space \cite{AY15, QLMG12,Gruska99,BK19}. An excellent and comprehensive survey on QFA was presented by  Ambainis and Yakaryilmaz \cite{AY15}. Moreover,  QFA have been studied significantly in physical experiment  \cite{MP20,  PHYF2022, TFLZZ2019}.

 One-way QFA (1QFA) were firstly proposed and studied by Moore and Crutchfield \cite{MC00}, Kondacs and Watrous \cite{KW97}, and then Ambainis and Freivalds \cite{AF98}, Brodsky and Pippenger \cite{BP02}, and other authors (e.g., the references in \cite{AY15, QLMG12,BK19}), where ``1" means ``one-way", that is, the tape-head moves only from left side to right side.  The decision problems regarding the equivalence of 1QFA and the minimization of states of 1QFA have been studied in  \cite{QLMG12, MQL12, QLZMG11, LQ06, LQ08}. 


More specifically, \textit{measure-once} one-way QFA (MO-1QFA) were initiated by Moore and Crutchfield \cite{MC00} and  \textit{measure-many } one-way QFA (MM-1QFA) were studied first by Kondacs and Watrous \cite{KW97}. In MO-1QFA, there is only a measurement for computing each input string, performing after reading the last symbol; in contrast, in MM-1QFA, measurement is performed after reading each symbol, instead of only the last symbol. 
Then other 1QFA   were also proposed and studied by Ambainis {\it et al.},  Nayak,  Hirvensalo,  Yakaryilmaz and Say, Paschen, Ciamarra,  Bertoni {\it et al.}, Qiu and Mateus {\it et al.} as well other authors (e.g., the references in \cite{AY15}), including:
Latvian QFA  \cite{ABG06},  QFA with control language \cite{BMP03},    1QFA with ancilla qubits (1QFA-A) \cite{Pas00},   one-way quantum finite automata together with classical states (1QFAC) \cite{QLMS15}, and other 1QFA such as  Nayak-1QFA (Na-1QFA), General-1QFA (G-1QFA), and fully 1QFA (Ci-1QFA), where G-1QFA, 1QFA-A, Ci-1QFA, 1QFA-CL, and 1QFAC can recognize all regular languages with bounded error.  For more details, we can refer to
 \cite{AY15, QLMG12, BK19}.



MO-1QFA have advantages over crisp finite automata in state complexity for recognizing some languages \cite{QLMG12, BK19}. Mereghetti and  Palano et al. \cite{MP20} realized an MO-1QFA with optic implementation and the state complexity of this MO-1QFA has exponential advantages over DFA and NFA as well as PFA \cite{Paz71}. 
MM-1QFA  have stronger computing power than MO-1QFA \cite{BP02}, but both MO-1QFA and MM-1QFA accept with bounded error only proper subsets  of regular languages. Indeed, Brodsky and Pippenger \cite{BP02} proved that the languages accepted by MO-1QFA with bounded error are exactly {\it reversible languages} that are accepted by {\it reversible finite automata} (RFA). RFA have three different definitions  and  were named as group automata, BM-reversible automata, and AF-reversible automata (see \cite{Qiu07}), respectively. In particular, these three definitions were proved to be equivalent in \cite{Qiu07}.



The remainder of the paper is organized as follows. In Section 2,  in the interest of readability, we first introduce basics in quantum computing, then one-way QFA are recalled and we focus on reviewing MO-1QFA and MM-1QFA.   The main contributions are in Sections 3 and 4. In Section 3, we first show the appropriate oracle to be used, that is, $\boldsymbol{AA}$ is not strong enough for learning RFA and MO-1QFA with polynomial time, and a more powerful oracle (named as $\boldsymbol{AD}$ oracle) is thus employed. With $\boldsymbol{AD}$ oracle we design an algorithm for learning MO-1QFA with polynomial time, and the correctness and complexity of algorithm are proved and analyzed in detail. Afterwards, in Section 4 we continue to design an algorithm for learning MM-1QFA with polynomial time.  Finally, the main results are summarized in Section 5, and further problems are mentioned for studying.

\section{ Preliminaries on  quantum computing and QFA}

 For the sake of readability, in this section we outline basic notations and principles in quantum computing and review the definitions of MO-1QFA, MM-1QFA,  and RFA. For more details, we can refer to \cite{NC00} and \cite{QLMG12,SY14,AY15,BK19}.

\subsection{ Basics in quantum computing}
Let $\mathbb{C}$ denote the set of all complex numbers, $\mathbb{R}$
 the set of all real numbers, and $\mathbb{C}^{n\times m}$
 the set of $n\times m$ matrices having entries in $\mathbb{C}$. Given two matrices $A\in \mathbb{C}^{n\times m}$ and $B\in\mathbb{C}^{p\times q}$, their {\it tensor product} is the $np\times
mq$ matrix, defined as \[A\otimes B=\begin{bmatrix}
  A_{11}B & \dots & A_{1m}B \\
  \vdots & \ddots & \vdots \\
  A_{n1}B &\dots & A_{nm}B \
\end{bmatrix}.\]
$(A\otimes B)(C\otimes D)=AC\otimes BD$ holds if the multiplication of matrices is satisfied.

 If $MM^\dagger=M^\dagger M=I$, then matrix $M\in\mathbb{C}^{n\times n}$ is  {\it
unitary},  where
 $\dagger$ denotes conjugate-transpose operation. $M$ is said to be {\it
Hermitian} if $M=M^\dagger$. For $n$-dimensional row vector
$x=(x_1,\dots, x_n)$, its norm  $||x||$ is defined as
$||x||=\big(\sum_{i=1}^n x_ix_i^{*}\big)^{\frac{1}{2}}$, where
symbol $*$ denotes conjugate operation.  Unitary operations preserve
the norm, i.e., $||xM||=||x||$ for each  $x\in \mathbb{C}^{1\times
n}$ and any unitary matrix $M\in\mathbb{C}^{n\times n}$.


According to the basic principles of quantum mechanics \cite{NC00}, a state of quantum system can be described by a unit vector in a Hilbert space. More specifically,
 let $B=\{q_1,\dots,
q_n\}$ associated with a quantum system denote a  basic state set, where every basic state $q_i$ can be represented by an
$n$-dimensional row vector $\langle q_i|=(0\dots1\dots0)$ having
only 1 at the $i$th entry (where $\langle \cdot|$ is
Dirac notation, i.e., bra-ket notation). At any time, the state of this system
is a {\it superposition} of these basic states and can be
represented by a row vector $\langle \phi|=\sum_{i=1}^nc_i\langle
q_i|$ with $c_i\in\mathbb{C}$ and $\sum_{i=1}^n|c_i|^2=1$; $|\phi\rangle$ represents the  conjugate-transpose of $\langle \phi|$. So, the quantum system is described by Hilbert space ${\cal H}_Q$ spanned by the base $\{|q_i\rangle: i=1,2,\dots,n\}$, i.e. ${\cal H}_Q=span\{| q_i\rangle: i=1,2,\dots,n\}$.

The state evolution of quantum system complies with unitarity.    Suppose the current state of system is $|\phi\rangle$. If it is acted on by some unitary matrix (or unitary operator) $M_1$,  then $|\phi\rangle$ is changed to the new current state $M_1|\phi\rangle$;  if the second unitary matrix, say $M_2$, is acted on  $M_1|\phi\rangle$, then $M_1|\phi\rangle$ is  changed to  $M_2 M_1|\phi\rangle$.
So, after  a series of unitary matrices $M_1, M_2, \ldots, M_k$ are performed in sequence, the system's state becomes $M_kM_{k-1}\cdots M_1|\phi\rangle$.

To get some information from the quantum system,  we need to make a measurement on its current state. Here we consider {\it projective measurement} (i.e. von Neumann measurement). A projective measurement is described by an {\it observable}  that is  a Hermitian matrix ${\cal O}=c_1P_1+\dots +c_s P_s$, where $c_i$ is its eigenvalue and, $P_i$ is the projector onto the eigenspace corresponding to $c_i$. In this case, the projective measurement of ${\cal O}$ has result set $\{c_i\}$ and projector set $\{P_i\}$. For example, given state $|\psi\rangle$ is made by the measurement ${\cal O}$, then the probability of obtaining result $c_i$ is $\|P_i|\psi\rangle\|^2$ and the state $|\psi\rangle$ collapses to $\frac{P_i|\psi\rangle}{\|P_i|\psi\rangle\|}$.

\subsection{Review of  one-way QFA and RFA}

For non-empty set $\Sigma$, by $\Sigma^{*}$ we mean the set of all finite length
strings over $\Sigma$, and $\Sigma^n$ denotes the set of all
strings over $\Sigma$ with length $n$. For $u\in \Sigma^{*}$,
$|u|$ is the length of $u$; for example, if
$u=x_{1}x_{2}\ldots x_{m}\in \Sigma^{*}$ where $x_{i}\in \Sigma$,
then $|u|=m$.  For set $S$,  $|S|$ denotes the cardinality of $S$.


\subsubsection{MO-1QFA}
We recall the definition of MO-1QFA. An MO-1QFA with $n$ states and input alphabet $\Sigma$ is a five-tuple \begin{equation}{\cal M}=(Q, |\psi_0\rangle,
\{U(\sigma)\}_{\sigma\in \Sigma}, Q_a,Q_r) \end{equation} where
\begin{itemize}
\item  $Q=\{|q_1\rangle,\dots,|q_n\rangle\}$ consist of  an orthonormal base that
spans a Hilbert space ${\cal H}_Q$ ($|q_i\rangle$ is identified with a column vector with  the $i$th entry 1 and the others 0);
at any time,
the state of ${\cal M}$ is a superposition of these basic states;

 \item $|\psi_0\rangle\in {\cal H}$ is the initial state;

\item for any $\sigma\in \Sigma$, $U(\sigma)\in \mathbb{C}^{n\times
n}$ is a unitary matrix;

\item $Q_a, Q_r\subseteq Q$ with $Q_a\cup Q_r=Q$ and $Q_a\cap Q_r=\emptyset$ are the accepting and rejecting states, respectively, and it describes an observable  by using the projectors
$P(a)=\sum_{|q_i\rangle\in Q_a}|q_i\rangle\langle q_i|$ and $P(r)=\sum_{|q_i\rangle\in Q_r}|q_i\rangle\langle q_i|$, with the result set $\{a,r \}$ of which `$a$'
and `$r$' denote ``accepting'' and ``rejecting'', respectively. Here $Q$ consists of accepting and rejecting sets.
\end{itemize}

Given an MO-1QFA ${\cal M}$ and an input word $s=x_1\dots
x_n\in\Sigma^{*}$, then starting from $|\psi_0\rangle$, $U(x_1),\dots,
U(x_n)$ are applied in succession, and at the end of the word, a
measurement $\{P(a),P(r)\}$ is performed with the result that ${\cal
M}$ collapses into accepting states or rejecting states with
corresponding probability. Hence, the probability $L_{\cal M}(x_1\dots x_n)$ of ${\cal M}$ accepting $w$ is defined as:
\begin{align}
L_{\cal M}(x_1\dots x_n)=\|P(a)U_s|\psi_0\rangle\|^2\label{f_MO}
\end{align}
where we denote $U_s=U_{x_n}U_{x_{n-1}}\cdots U_{x_1}$.


\subsubsection{RFA}

Now we recollect reversible finite automata (RFA). As mentioned above, there are three equivalent definitions for RFA \cite{Qiu07}, that is  group automata,  BM-reversible automata, and AF-reversible automata. Here we describe group automata. First we review DFA. A DFA  $G=(S, s_0, \Sigma, \delta, S_a)$, where $S$ is a finite state set, $s_0\in S$ is its  initial state, $S_a\subseteq  S$ is its accepting state set, $\Sigma$ is and input alphabet,  and $\delta$ is a transformation function, i.e., a mapping $\delta: S\times \Sigma\rightarrow S$.

 An RFA (group automaton) $G=(S, s_0, \Sigma, \delta, S_a)$ is DFA  and satisfies that for any $q\in S$ and any $\sigma\in\Sigma$, there is unique $p\in S$ such that $\delta(p,\sigma)=q$.

 The languages accepted by MO-1QFA with bounded error is exactly the languages accepted by RFA \cite{BP02}.  In fact, RFA are the special cases of MO-1QFA, and this is showed by the following proposition.

\begin{proposition}

(1) For any MO-1QFA ${\cal M}=(Q, |\psi_0\rangle,
\{U(\sigma)\}_{\sigma\in \Sigma}, Q_a,Q_r) $ with $|\psi_0\rangle\in Q$, if all entries in $U(\sigma)$ for each $\sigma\in \Sigma$ are either $0$ or $1$, then ${\cal M}$ is actually a group automaton.  (2) If $G=(S, s_0, \Sigma, \delta, S_a)$  is  a group automaton, then $G$ is actually an  MO-1QFA.
\end{proposition}

\begin{proof} (1)
Suppose the base states $Q=\{|q_i\rangle: i=1,2,\ldots,n\}$, where $|q_i\rangle$ is an $n$-dimensional column vector with  the $i$th entry 1 and the others 0. Let $|\psi_0\rangle=|q_{i_0}\rangle$  for some $i_0\in \{1,2,\ldots,n\}$. It is clear that $U(\sigma)$ (for each $\sigma\in\Sigma$) is a permutation matrix and therefore $U(\sigma)$ is also a bijective mapping from $Q$ to $Q$. So, ${\cal M}$ is  a group automaton.

(2)      If $G=(S, s_0, \Sigma, \delta, S_a)$  is  a group automaton with $|S|=n$, then  we denote $S=\{q_1,q_2,\ldots, q_n\}$ and $s_0$ is some $q_i\in S$. According to the definition of group automata, for each  $\sigma\in\Sigma$, $\delta(\cdot,\sigma)$ is a bijective mapping from $S$ to $S$.  So, we can identify $q_i$ with an $n$-dimensional column vector with  the $i$th entry 1 and the others 0. Then for each  $\sigma\in\Sigma$, $\delta(\cdot,\sigma)$ induces a unitary matrix $U(\sigma)$ acting  on the $n$-dimensional Hilbert space  spanned by the base states $S$. Finally, $S_a\subseteq S$ and $S_r=S\setminus S_a$  are accepting and rejecting sets of states, respectively. As a result, $G$ is actually equivalent to an MO-1QFA.

\end{proof}


\subsubsection{MM-1QFA}
We review the definition of MM-1QFA. Formally,  given an input alphabet $\Sigma$ and an end-maker
$\$\notin\Sigma$, an
 MM-QFA with $n$ states over the {\it working
 alphabet}
 $\Gamma=\Sigma\cup\{\$\}$ is a six-tuple
 ${\cal M}=(Q, |\psi_0\rangle,
\{U(\sigma)\}_{\sigma\in \Sigma}, Q_a,Q_r,Q_g)$, where
\begin{itemize}
\item  $Q$, $|\psi_0\rangle$, and $U(\sigma)$ ($\sigma\in \Gamma$) are
defined as in the case of MO-1QFA,  $Q_a,Q_r,Q_g$ are disjoint to each other and represent the accepting, rejecting, and going states, respectively.

\item The measurement is  described by the projectors
$P(a)$, $P(r)$ and $P(g)$, with the results in $\{a,r,g \}$ of
which `$a$', `$r$' and `$g$' denote ``accepting'', ``rejecting'' and
``going on'', respectively.
\end{itemize}

Any input word $w$ to MM-1QFA is in the form: $w\in\Sigma^{*}\$$,
with symbol $\$$ denoting the end of a word. Given an input word
$x_1\dots x_n\$$ where $x_1\dots x_n\in \Sigma^n$, MM-1QFA ${\cal
M}$ performs the following computation:
\begin{itemize}
\item [1.] Starting from $|\psi_0\rangle$, $U(x_1)$ is applied, and then we
get a new state $|\phi_1\rangle= U(x_1)|\psi_0\rangle$. In succession, a
measurement of ${\cal O}$ is performed on $|\phi_1\rangle$, and
then the measurement result $i$ ($i\in \{a,g,r\}$) is yielded as
well as  a new state
$|\phi_1^{i}\rangle=\frac{P(i)|\phi_1\rangle}{\sqrt{p_1^i}}$
 is obtained, with corresponding probability
 $p_1^i=||P(i)|\phi_1\rangle||^2$.

 \item [2.]   In the above step, if $|\phi_1^{g}\rangle$ is obtained, then  starting from $|\phi_1^{g}\rangle$, $U(x_2)$ is
 applied and a measurement $\{P(a),P(r),P(g)\}$ is performed. The
  evolution rule is  the same as the above step.

 \item [3.] The process continues as far as the measurement result
 `$g$' is yielded. As soon as the measurement result is `$a$'(`$r$'), the
 computation halts and the input word is accepted (rejected).
\end{itemize}

Thus, the probability $L_{\cal M}(x_1\dots x_n)$ of ${\cal M}$ accepting $w$ is defined as:
\begin{align} \label{QFAP1}
&L_{\cal M}(x_1\dots
x_n)\\
=&\sum^{n+1}_{k=1}||P(a)U(x_k) \prod^{k-1}_{i=1}\big(P(g)U(x_i)\big) |\psi_0\rangle ||^2,
\end{align}
or equivalently,
\begin{align}\label{QFAP2}
&L_{\cal M}(x_1\dots
x_n)\\
=&\sum^{n}_{k=0}||P(a)U(x_{k+1}) \prod^{k}_{i=1}\big(P(g)U(x_i)\big) |\psi_0\rangle ||^2,
\end{align}
where, for simplicity, we can denote $\$$ by $x_{n+1}$  if no confusion results.

\section{Learning MO-1QFA}

First we recall a definition concerning model learning  with an oracle in polynomial time.

\begin{definition}\cite{Tze92}
Let \textbf{R} be a class to be learned and $O_{{\bf R}}$ be an oracle for \textbf{R}. Then \textbf{R} is said
to be polynomially learnable using the oracle $O_{{\bf R}}$ if there is a learning algorithm $L$ and
a two-variable polynomial $p$ such that for every target $r \in \textbf{R}$ of size $n$ to be learned, $L$ runs in time
$p(n, m)$ at any point and outputs a hypothesis that is equivalent to $r$, where $m$ is the maximum
length of data returned by $O_{{\bf R}}$ so far in the run.

\end{definition}

In order to learn a model with polynomial time via an oracle, we hope this oracle is as weaker as possible. For learning MO-1QFA, suppose an oracle can only answer the amplitudes of accepting states for each input string, then can we learn MO-1QFA successfully with polynomial time via such an oracle?  We name such an oracle as $\boldsymbol{AA}$ oracle. For clarifying this point, we try to use $\boldsymbol{AA}$ oracle to learning DFA. In this case,  $\boldsymbol{AA}$ oracle can answer if it is either an accepting state or a rejecting state for each input string. Equivalently, $\boldsymbol{AA}$ oracle is exactly  membership query for learning DFA as the target model. Therefore, learning DFA via $\boldsymbol{AA}$ oracle is not polynomial by virtue of the following Angluin's result \cite{Ang88}.

\begin{theorem}\cite{Ang88}
 DFA are not polynomially learnable using the membership  query oracle only.

\end{theorem}

In fact, in 2007 a stronger result  was proved in \cite{TK07} that 0-reversible automata (that is, a  0-reversible automaton is defined as a RFA with only one accepting state \cite{Ang82})  are not learnable by using membership query  only.  That can be described by the following theorem.



\begin{theorem}\cite{TK07}\label{TK07}
Any  RFA with only one accepting state is not learnable by using membership queries  only.

\end{theorem}

Therefore, we have the following proposition.




\begin{proposition}

DFA and RFA as well as MO-1QFA are not  learnable using $\boldsymbol{AA}$ oracle only.

\end{proposition}

\begin{proof}

Due to the above Theorem \ref{TK07}, we know that any RFA  is not learnable by using membership query  only. 
Since RFA are special cases of DFA and MO-1QFA, we obtain that neither DFA nor MO-1QFA is  learnable by using membership query  only. 

For learning DFA and RFA,  $\boldsymbol{AA}$ oracle  is exactly equal to membership  query oracle, so   we  conclude that DFA and MO-1QFA are not  learnable using $\boldsymbol{AA}$ oracle only.

\end{proof}

So, we consider a stronger oracle, named as    $\boldsymbol{AD}$   oracle that can answer all amplitudes (instead of the amplitudes of accepting states only)  of the superposition state for each input string. For example, for quantum state $|\psi\rangle=\sum_{i=1}^{n}\alpha_i|q_i\rangle$ where $\sum_{i=1}^{n}|\alpha_i|^2=1$,  $\boldsymbol{AA}$ oracle can only answer the amplitudes of accepting states in $\{|q_1\rangle, |q_2\rangle, \ldots, |q_n\rangle\}$, but $\boldsymbol{AD}$   oracle can answer the amplitudes for all states in $\{|q_1\rangle, |q_2\rangle, \ldots, |q_n\rangle\}$.  Using   $\boldsymbol{AD}$   oracle, we can prove that MO-1QFA and MM-1QFA are polynomially learnable. Therefore, for learning DFA or RFA, $\boldsymbol{AD}$   oracle can answer a concrete state for each input string, where the concrete state is the output state of the target automaton to be learned.

First we can easily prove that RFA are linearly learnable via  $\boldsymbol{AD}$   oracle, and this is the following proposition.

 \begin{proposition} \label{LRFA}
Let RFA  $G=(S, s_0, \Sigma, \delta, S_a)$ be the target to be learned. Then $G$ is linearly  learnable via using $\boldsymbol{AD}$ oracle with query complexity at most $|S||\Sigma|$.
\end{proposition}

 \begin{proof}
First, $\boldsymbol{AD}$ oracle can answer the initial state $s_0$ via inputting empty string. Then by taking $s_0$  as a vertex, we use pruning algorithm of decision tree to obtain all states in $G$ while accepting states are marked as well. It is easy to know that the query complexity is $O(|S||\Sigma|)$.


\end{proof}



 Our main concern is whether $\boldsymbol{AD}$   oracle is  too strong, that is to say, whether $\boldsymbol{AD}$   oracle  possesses too much information for our learning tasks.
 To clarify this point partially,  we employ a consistency problem that, in a way, demonstrates  $\boldsymbol{AD}$   oracle is really not too strong for our  model learning if the time complexity is polynomial.  So, we first recall an outcome from \cite{Tze92}.

 \begin{theorem}\cite{Tze92}
For any alphabet $\Sigma$ and finite set $S=\{q_1,q_2,\ldots,q_n\}$,   the following   problem is NP-complete: Given a set $D$  with $D\subseteq \Sigma^*\times S$, determine whether there is a DFA $G=(S_1, s_0, \Sigma, \delta, S_a)$  such that  for any $(x,q)\in D$, $\delta(s_0,x)=q$, where $p\in S_1$ if and only if  $ (x,p)\in D $ for some $x\in \Sigma^*$.

\end{theorem}

\begin{remark}
In above Theorem, each element in $D$ consists of a string in $\Sigma^*$ and a state in $S$, so $D$ can be identified with the information  carried by $\boldsymbol{AD}$   oracle in order to learn a DFA. That is to say, even if an $\boldsymbol{AD}$   oracle holds so much information like $D$ contained in a DFA in this way,  it is still not easy (NP-complete) to learn a consistent DFA.  To a certain extent,  $\boldsymbol{AD}$   oracle  is not too strong to learn a DFA.  Since constructing RFA is not easier than constructing DFA and RFA are special MO-1QFA,  we use  $\boldsymbol{AD}$  oracle for learning  MO-1QFA and MM-1QFA.  However, we still do not know what is the weakest oracle to learn MO-1QFA and MM-1QFA with polynomial-time query complexity.

\end{remark}











Let ${\cal M}=(Q, |\psi_0\rangle,
\{U(\sigma)\}_{\sigma\in \Sigma}, Q_a,Q_r) $ be the target MO-1QFA to be learned, where, as the case of learning PFA \cite{Tze92},  the learner is supposed to have the information of $Q, Q_a, Q_r$, but the other parameters are to be learned by mean of querying the oracle for achieving an equivalent MO-1QFA (more concretely, for each $\sigma\in\Sigma$, unitary matrix $V(\sigma)$ corresponding to $U(\sigma)$ needs to be determined, but it is  possible that $V(\sigma)\not=U(\sigma)$). For any $x\in\Sigma^*$, $\boldsymbol{AD}(x)$ can answer an amplitude distribution that is exactly equivalent to a state of superposition corresponding to the input string $x$, more exactly, $\boldsymbol{AD}(x)$ can answer the same state as $U(\sigma_k)U(\sigma_{k-1})\cdots U(\sigma_1)|\psi_0\rangle$ where $x=\sigma_1\sigma_2\cdots\sigma_k$. From now on,  we denote $U(x)=U(\sigma_k)U(\sigma_{k-1})\cdots U(\sigma_1)$ for $x=\sigma_1\sigma_2\cdots\sigma_k$.

We outline the basic idea and method for designing the learning algorithm of MO-1QFA  ${\cal M}$. First, the initial state can be learned from $\boldsymbol{AD}$ oracle by querying empty string $\varepsilon$. Then by using $\boldsymbol{AD}$ oracle we continue to search for a base of the Hilbert space spanned by $\{v^*=U(x)|\psi_0\rangle: x\in\Sigma^*\}$. This procedure will be terminated since the dimension of the space is at most $|Q|$. In fact, we can prove this can be finished in polynomial time. Finally, by virtue of the learned base and solving groups of linear equations we can conclude  $V(\sigma)$ for each $\sigma\in\Sigma$. We prove these results in detail following the algorithm, and now present Algorithm 1 for learning MO-1QFA as follows.

\begin{algorithm}
\caption{Algorithm for learning MO-1QFA  ${\cal M}=(Q, |\psi_0\rangle,
\{U(\sigma)\}_{\sigma\in \Sigma}, Q_a,Q_r) $}\label{LMO}
\begin{algorithmic}[1]
    \State $|\psi_0^*\rangle\leftarrow \boldsymbol{AD}(\varepsilon);$
    \State Set $\mathcal{B}$ to be the empty set$;$
    \State Set $Nod\leftarrow\{node(\varepsilon)\};$
    \While {$Nod$ is not empty}
    \State $\mathbf{begin}$ Take an element $node(x)$ from $Nod$$;$
           \State $v^*(x)\leftarrow \boldsymbol{AD}(x);$
               \If {$v^*(x) \notin span(\mathcal{B})$}
                \State{$\mathbf{begin}$ Add $node(x\sigma)$  to $Nod$ for all  $\sigma\in \Sigma$$;$}
                    \State{\qquad $\mathcal{B}\leftarrow \mathcal{B}\cup\{v^*(x)\}\ \mathbf{end};$}
                \EndIf
    \EndWhile
    \State $\mathbf{end};$
    \State Let $V(\sigma)=[x_{ij}(\sigma)]$ for any $\sigma\in \Sigma$,$\ 1\leq i,j\leq n;$
    \State Define a linear system:
    \State $\qquad$ for any $ v^*(x) \in \mathcal{B}$ and any $\sigma\in\Sigma$,$\ V(\sigma)v^*(x)=v^*(x\sigma)=\boldsymbol{AD}(x\sigma)$$,$
    \State $\qquad$ for $ 1\leq i_1,i_2\leq n$,  if $i_1\neq i_2$, then $\ \sum_{j=1}^nx_{i_1j}(\sigma) \overline{x_{i_2j}(\sigma)  }=0$; otherwise, it is 1$,$
    \State Find a suitable solution for $x_{ij}(\sigma)$'s,
    \If {there is a solution}
    \State return ${\cal M}^*=(Q, |\psi_0^*\rangle,
\{V(\sigma)\}_{\sigma\in \Sigma}, Q_a,Q_r)$
    \Else
    \State return (not exist)$;$
    \EndIf
\end{algorithmic}
\end{algorithm}


Next we prove the correctness of Algorithm 1 and then analyze its complexity. First we prove that Step 1 to Step 12 in Algorithm 1 can produce a set of vectors $\mathcal{B}$ consisting of a base of space spanned by $\{v^*(x)| x\in\Sigma^*\}$, where $v^*(x)=\boldsymbol{AD}(x)$ is actually  the vector replied by oracle $\boldsymbol{AD}$ for input string $x$, that is $v^*(x)=\boldsymbol{AD}(x)=U(\sigma_k)U(\sigma_{k-1})\ldots U(\sigma_1)|\psi_0\rangle$, for $x=\sigma_1\sigma_2\ldots\sigma_k$.

\begin{proposition} \label{Algorithm1P}

In Algorithm 1 for learning MO-1QFA,  the final  set of vectors $\mathcal{B}$ consists of a base of Hilbert space $span\{v^*(x)| x\in\Sigma^*\}$ that is spanned by $\{v^*(x)| x\in\Sigma^*\}$.

\end{proposition}

\begin{proof}
From the algorithm procedure we can assume that ${\cal B}=\{v^*(x_1), v^*(x_2),\ldots,v^*(x_m)\}$ for some $m$, where it is clear that some $x_i$ equals to $\varepsilon$, and for any $x\in\Sigma^*$, there are  $x_j$ and $y\in\Sigma^*$ such that $x=x_jy$. The rest is to show that $v^*(x)$ can be linearly represented by the vectors in ${\cal B}$ for any $x\in\Sigma^*$. Let $x=x_jy$ for some $x_j$ and $y\in\Sigma^*$. By induction on the length $|y|$ of $y$. If $|y|=0$, i.e., $y=\varepsilon$, then it is clear for $x=x_j$. If $|y|=1$, then due to the procedure of algorithm, $v^*(x_jy)$ is linearly dependent on ${\cal B}$.  Suppose that it holds for $|y|=k\geq 0$. Then we need to verify it holds for $|y|=k+1$. Denote $y=z\sigma$ with $|z|=k$. Then with induction hypothesis we have $v^*(x_jz)= \sum_{k}c_kv^*(x_k)$. Therefore we have
\begin{align}
v^*(x)&=v^*(x_jz\sigma)\nonumber\\
&=U(\sigma)v^*(x_jz)\nonumber\\
&=U(\sigma)\sum_{k}c_kv^*(x_k)\nonumber\\
&=\sum_{k}c_kv^*(x_k\sigma).
 \end{align}
Since $v^*(x_k\sigma)$ is linearly dependent on ${\cal B}$ for $k=1,2,\ldots,m$, the proof is completed.

\end{proof}

The purpose of Algorithm 1 is to learn the target MO-1QFA ${\cal M}=(Q, |\psi_0\rangle,
\{U(\sigma)\}_{\sigma\in \Sigma}, Q_a,Q_r)$, so we need to verify ${\cal M}^*=(Q, |\psi_0^*\rangle,
\{V(\sigma)\}_{\sigma\in \Sigma}, Q_a,Q_r)$ obtained is equivalent to ${\cal M}$. For this  it suffices to check  $V(x)|\psi_0^*\rangle=U(x)|\psi_0\rangle$ for any $x\in\Sigma^*$,
 where $V(x)=V(\sigma_s)V(\sigma_{s-1})\ldots V(\sigma_1)$ and $U(x)=U(\sigma_s)U(\sigma_{s-1})\ldots U(\sigma_1)$ for $x=\sigma_1\sigma_2\ldots\sigma_s$.

\begin{theorem}\label{Algorithm1T}

In Algorithm 1 for learning MO-1QFA, for any $x\in\Sigma^*$,
 \begin{equation}
 V(x)|\psi_0^*\rangle=U(x)|\psi_0\rangle.
 \end{equation}

\end{theorem}

\begin{proof}

For $x=\varepsilon$, $V(\varepsilon)=U(\varepsilon)=I$, and $|\psi_0^*\rangle=\boldsymbol{AD}(\varepsilon)=|\psi_0\rangle$, so it holds.

For any $\sigma\in\Sigma$ and for any $ v^*(x) \in \mathcal{B}$, according to  Algorithm 1, we have $\ V(\sigma)v^*(x)=v^*(x\sigma)=\boldsymbol{AD}(x\sigma)$. In particular, taking $x=\varepsilon$, then we have
$\ V(\sigma)|\psi_0^*\rangle=v^*(\sigma)=\boldsymbol{AD}(\sigma)=U(x)|\psi_0\rangle$.

Suppose it holds for $|x|= k$. The rest is to prove that it holds for $|x|= k+1$. Denote $y=x\sigma$ where $|x|= k$ and $\sigma\in\Sigma$. Due to Proposition \ref{Algorithm1P}, $v^*(x)$ can be linearly represented by ${\cal B}=\{v^*(x_1), v^*(x_2),\ldots,v^*(x_m)\}$, i.e., $v^*(x)=\sum_{k}c_kv^*(x_k)$ for some $c_k\in \mathbb{C}$. With the induction hypothesis,  $V(x)|\psi_0^*\rangle=U(x)|\psi_0\rangle$ holds.
Then by means of Algorithm 1 we have
\begin{align}
V(y)|\psi_0^*\rangle &=V(x\sigma)|\psi_0^*\rangle \nonumber\\
&=V(\sigma)V(x)|\psi_0^*\rangle \nonumber\\
&=V(\sigma)U(x)|\psi_0\rangle \nonumber\\
&=V(\sigma)v^*(x)\nonumber\\
&=V(\sigma)\sum_{k}c_kv^*(x_k)\nonumber\\
&=\sum_{k}c_kV(\sigma)v^*(x_k)\nonumber\\
&=\sum_{k}c_kv^*(x_k\sigma).\nonumber\\
 \end{align}
On the other hand, since $v^*(z)=
\boldsymbol{AD}(z)=U(z)|\psi_0\rangle$ for any $z\in \Sigma^*$, we have

\begin{align}
\sum_{k}c_kv^*(x_k\sigma)
&=\sum_{k}c_kU(x_k\sigma)|\psi_0\rangle\nonumber\\
&=\sum_{k}c_kU(\sigma)U(x_k)|\psi_0\rangle\nonumber\\
&=\sum_{k}c_kU(\sigma)v^*(x_k)\nonumber\\
&=U(\sigma)\sum_{k}c_kv^*(x_k)\nonumber\\
&=U(\sigma)v^*(x)\nonumber\\
&=U(x\sigma)|\psi_0\rangle\nonumber \\
&=U(y)|\psi_0\rangle.
 \end{align}
So, the proof is completed.

\end{proof}

From Theorem \ref{Algorithm1T} it follows that Algorithm 1 returns an equivalent MO-1QFA to the target MO-1QFA ${\cal M}=(Q, |\psi_0\rangle,
\{U(\sigma)\}_{\sigma\in \Sigma}, Q_a,Q_r)$ to be learned. Next we analyze the computational complexity of Algorithm 1.

\begin{proposition} \label{Algorithm1:Complexity}

Let the target MO-1QFA to be learned have $n$'s bases states. The the computational complexity of Algorithm 1 is  $O(n^5|\Sigma|)$.

\end{proposition}

\begin{proof}

We consider it from two parts.

(I) The first part of Algorithm 1 to get ${\cal B}$: The complexity to determine the linear independence of some $n$-dimensional vectors is $O(n^3)$ \cite{FF63}, and there are at most $n$ time to check this, so the first part of Algorithm 1 to get ${\cal B}$ needs time $O(n^4)$.

(II) The second part of finding the feasible solutions for  $V(\sigma)$ for each $\sigma\in\Sigma$: For any $\sigma\in\Sigma$, Step 15 defines $|{\cal B}|$'s matrix equations and these equations are clearly equivalent to a group of linear equations, but are subject to the restriction conditions in Step 16. So, this part is actually a problem of linear programming and we can refer to \cite {BV04,Ka84} to get the time complexity is $O(n^5|\Sigma|)$.

Therefore, by combining (I) and (II) we have the complexity of Algorithm 1 is $O(n^5|\Sigma|)$.

\end{proof}

To illustrate Algorithm 1 for learning MO-1QFA, we give an example as follows.

\begin{example} Suppose that  $\mathcal{M}=(Q,|\psi_0\rangle,\{U(\sigma)\}_{\sigma\in\Sigma},Q_{a},Q_r)$ is an ${\rm MO}$-${\rm 1QFA} $ to be learned by Algorithm 1, where $Q=\{q_{0},q_{1}\}$, $Q_a=\{q_{1}\}$, $Q_r=\{q_{0}\}$, $ \Sigma=\{a\}$, $U(a)=\begin{bmatrix}
\frac{1}{\sqrt{2}} &  \frac{1}{\sqrt{2}}\\
\frac{1}{\sqrt{2}} &  -\frac{1}{\sqrt{2}}\\
\end{bmatrix}$, states $q_0$ and $q_1$ correspond to the two quantum basis states
$|q_0\rangle=\begin{bmatrix}
1\\
0 \\
\end{bmatrix}$ and $|q_1\rangle=\begin{bmatrix}
0\\
1 \\
\end{bmatrix}$, and $|\psi_0\rangle=|q_0\rangle$.
Denote $\mathcal{M}^*=(Q,|\psi_0\rangle^*,\{V(\sigma)\}_{\sigma\in\Sigma},Q_{a},Q_r)$ as the ${\rm MO}$-${\rm 1QFA} $  learned from Algorithm 1, and the procedure for obtaining $\mathcal{M}^*$ from Algorithm 1 is given below.

Step 1 of Algorithm 1 yields
\begin{equation}
\begin{split}
|\psi_0\rangle^*=&\textbf{AD}(\epsilon)\\
=&U(\epsilon)|\psi_0\rangle\\
=&I|\psi_0\rangle\\
=&|q_0\rangle.
\end{split}
\end{equation}

The 1st iteration run of the while loop body in Algorithm 1 is given below, with the computation of each set.

Step 6 of Algorithm 1 yields
\begin{equation}
\begin{split}
v^*(\epsilon)=&\textbf{AD}(\epsilon)\\
=&U(\epsilon)|\psi_0\rangle\\
=&I|\psi_0\rangle\\
=&|q_0\rangle.
\end{split}
\end{equation}

Step 8 of Algorithm 1 yields
\begin{equation}
\begin{split}
Nod=\{node(a)\}.
\end{split}
\end{equation}

Step 9 of Algorithm 1 yields
\begin{equation}
\begin{split}
\mathcal{B}=\{|q_0\rangle\}.
\end{split}
\end{equation}

The 2nd iteration run of the while loop body in Algorithm 1 is given below, with the computation of each set.

Step 6 of Algorithm 1 yields
\begin{equation}
\begin{split}
v^*(a)=&\textbf{AD}(a)\\
=&U(a)|\psi_0\rangle\\
=&U(a)|q_0\rangle\\
=&\frac{|q_0\rangle+|q_1\rangle}{\sqrt{2}}.
\end{split}
\end{equation}

Step 8 of Algorithm 1 yields
\begin{equation}
\begin{split}
Nod=\{node(aa)\}.
\end{split}
\end{equation}

Step 9 of Algorithm 1 yields
\begin{equation}
\begin{split}
\mathcal{B}=\left\{|q_0\rangle,\frac{|q_0\rangle+|q_1\rangle}{\sqrt{2}}\right\}.
\end{split}
\end{equation}

The 3rd iteration run of the while loop body in Algorithm 1 is given below, with the computation of each set.

Step 6 of Algorithm 1 yields
\begin{equation}
\begin{split}
v^*(aa)=&\textbf{AD}(aa)\\
=&U(aa)|\psi_0\rangle\\
=&U(a)U(a)|q_0\rangle\\
=&|q_0\rangle.
\end{split}
\end{equation}

Since $v^*(x)$ belongs to $span(\mathcal{B})$, the statements in the branch statement are not executed at this point. The set $Nod$ is the empty set at this point, so Algorithm 1 exits from the while loop body.

Finally, let 
$V(a)=\begin{bmatrix}
x_{11}(a) &  x_{12}(a) \\
x_{21}(a)  &  x_{22}(a) \\
\end{bmatrix}$, and according  to steps 15 and 16 of  Algorithm 1, we get
\begin{equation}\label{exampleeq1}
\begin{split}
V(a)v^*(\epsilon)=&V(a)|q_0\rangle\\
=&v^*(a)\\
=&\textbf{AD}(a)\\
=&\frac{|q_0\rangle+|q_1\rangle}{\sqrt{2}}.
\end{split}
\end{equation}

\begin{equation}\label{exampleeq2}
\begin{split}
V(a)v^*(a)=&V(a)\left(\frac{|q_0\rangle+|q_1\rangle}{\sqrt{2}}\right)\\
=&v^*(aa)\\
=&\textbf{AD}(aa)\\
=&|q_0\rangle.
\end{split}
\end{equation}

From Eq. (\ref{exampleeq1}) and Eq. (\ref{exampleeq2}), the following system of equations is obtained
\begin{equation}\label{exampleeq3}
\begin{cases}
\begin{bmatrix}
x_{11}(a) &  x_{12}(a) \\
x_{21}(a)  &  x_{22}(a) \\
\end{bmatrix}
\begin{bmatrix}
1 \\
0 \\
\end{bmatrix}=\begin{bmatrix}
\frac{1}{\sqrt{2}}\\
\frac{1}{\sqrt{2}} \\
\end{bmatrix},\\
\begin{bmatrix}
x_{11}(a) &  x_{12}(a) \\
x_{21}(a)  &  x_{22}(a) \\
\end{bmatrix}
\begin{bmatrix}
\frac{1}{\sqrt{2}}\\
\frac{1}{\sqrt{2}} \\
\end{bmatrix}=\begin{bmatrix}
1 \\
0\\
\end{bmatrix}.
\end{cases}
\end{equation}

Solving the system of Eq. (\ref{exampleeq3}) gives
\begin{equation}
\begin{split}
V(a)=&\begin{bmatrix}
x_{11}(a) &  x_{12}(a) \\
x_{21}(a)  &  x_{22}(a) \\
\end{bmatrix}\\
=&\begin{bmatrix}
\frac{1}{\sqrt{2}} &  \frac{1}{\sqrt{2}}\\
\frac{1}{\sqrt{2}} &  -\frac{1}{\sqrt{2}}\\
\end{bmatrix}.
\end{split}
\end{equation}

As a result, we can obtain $\mathcal{M}^*=(Q,|\psi_0\rangle^*,\{V(\sigma)\}_{\sigma\in\Sigma},Q_{a},Q_r)$, where $Q=\{q_{0},q_{1}\}$, $|\psi_0\rangle^*=|q_0\rangle$, $Q_a=\{q_{1}\}$, $Q_r=\{q_{0}\}$, $ \Sigma=\{a\}$,  and $V(a)=\begin{bmatrix}
\frac{1}{\sqrt{2}} &  \frac{1}{\sqrt{2}}\\
\frac{1}{\sqrt{2}} &  -\frac{1}{\sqrt{2}}\\
\end{bmatrix}$.

Therefore, the ${\rm MO}$-${\rm 1QFA} $ $\mathcal{M}^*$, which is equivalent to ${\rm MO}$-${\rm 1QFA} $ $\mathcal{M}$, can be obtained from Algorithm 1.

\end{example}

\section{Learning MM-1QFA}

In this section, we study learning MM-1QFA via  $\boldsymbol{AD}$ oracle. Let ${\cal M}=(Q, |\psi_0\rangle, \{U(\sigma)\}_{\sigma\in \Gamma}, Q_a,Q_r,Q_g)$ be the target QFA to be learned, where $\Gamma=\Sigma\cup\{\$\}$, and $\$\notin\Sigma$ is an end-maker.  As usual,  $Q, Q_a, Q_r, Q_g$ are supposed to be known, and the goal is to achieve unitary matrices $V(\sigma)$ for each $\sigma\in \Gamma$   in order to get an equivalent MM-1QFA
${\cal M}^*=(Q, |\psi_0^*\rangle, \{V(\sigma)\}_{\sigma\in \Gamma}, Q_a,Q_r,Q_g)$. $\boldsymbol{AD}$ oracle can answer an amplitude distribution   $\boldsymbol{AD}(x)$    for any $x\in\Gamma^*$.  MM-1QFA performs measuring after reading each input symbol, and then only the non-halting (i.e. going on) states continues to implement computing for next step, and  
the amplitude distribution for the superposition state after performing each unitary matrix needs to be learned from oracle.

 Therefore, for any $x=\sigma_1\sigma_2\ldots\sigma_k\in\Gamma^*$,  since MM-1QFA  ${\cal M}$ outputs the following state (un-normalized form) as the current state:
\begin{equation}
U(\sigma_k)P_nU(\sigma_{k-1})P_n\ldots U(\sigma_1)P_n|\psi_0\rangle,
\end{equation}
 we require  $\boldsymbol{AD}$ oracle can answer $\boldsymbol{AD}(x)=U(\sigma_k)P_nU(\sigma_{k-1})P_n\ldots U(\sigma_1)P_n|\psi_0\rangle$. In particular, $\boldsymbol{AD}(\varepsilon)=|\psi_0\rangle$.

Before presenting the algorithm of learning MM-1QFA,  we describe the main ideas and procedure.

First the initial state can be learned from $\boldsymbol{AD}$ oracle via querying empty string $\varepsilon$.

Then by using $\boldsymbol{AD}$ oracle we are going to search for a base $\mathcal{B}$ of the Hilbert space spanned by $\{v^*(x): x\in\Sigma^*\}$ where for any $x=\sigma_1\sigma_2\ldots\sigma_k\in\Sigma^*$,
\begin{equation}
v^*(x)=\boldsymbol{AD}(x)=U(\sigma_k)P_nU(\sigma_{k-1})P_n\ldots U(\sigma_1)P_n|\psi_0\rangle.
\end{equation}
 This procedure will be terminated due to the finite dimension of the space (at most $|Q|$), and  this can be completed with polynomial time.

 Finally, by combining  the base $\mathcal{B}$   and with  groups of linear equations we can obtain $V(\sigma)$ for each $\sigma\in\Sigma$. These results can be verified in detail after Algorithm 2, and we now present Algorithm 2 for learning MM-1QFA in the following.

\begin{algorithm}
\caption{Algorithm for learning MM-1QFA ${\cal M}=(Q, |\psi_0\rangle, \{U(\sigma)\}_{\sigma\in \Gamma}, Q_a,Q_r,Q_g)$}
\begin{algorithmic}[1]
    \State $|\psi_0^*\rangle\leftarrow \boldsymbol{AD}(\varepsilon);$ $v^*(\$)\leftarrow \boldsymbol{AD}(\$);$
     \State Set $\mathcal{B}$  to be the empty sets$;$
    \State Set $Nod\leftarrow\{node(\varepsilon)\};$
            \While {$Nod$s not empty}
    \State $\mathbf{begin}$ Take an element $node(x)$ from $Nod$;
           \State $v^*(x)\leftarrow \boldsymbol{AD}(x);$ 
                          \If {$v^*(x) \notin span(\mathcal{B})$}
                                          \State{$\mathbf{begin}$ Add $node(x\sigma)$  to $Nod$ for all  $\sigma\in \Sigma$$;$}
                    \State{\qquad $\mathcal{B}\leftarrow \mathcal{B}\cup\{v^*(x)\}    \ \mathbf{end};$}
                \EndIf
    \EndWhile
    \State $\mathbf{end};$
    \State Let $V(\sigma)=[x_{ij}(\sigma)]$ for any $\sigma\in \Sigma\cup \{\$\}$,$\ 1\leq i,j\leq n;$
    \State Define linear systems:
    \State $\qquad$ for any $ v^*(x) \in \mathcal{B}$ and  any $\sigma\in\Sigma\cup \{\$\}$,$\ V(\sigma)P_nv^*(x)=v^*(x\sigma)=\boldsymbol{AD}(x\sigma)$;
    \State $\qquad$ for $ 1\leq i_1,i_2\leq n$,  if $i_1\neq i_2$, then $\ \sum_{j=1}^nx_{i_1j}(\sigma) \overline{x_{i_2j}(\sigma)  }=0$; otherwise, it is 1$,$
    \State Find a suitable solution for $x_{ij}(\sigma)$'s, and denote $\mathbb{V}=\{V(\sigma):\sigma\in\Sigma\cup\{\$\}\}$$;$
    \If {there is a solution}
    \State return $\mathcal{M}^*=(Q, |\psi_0^*\rangle, \{V(\sigma)\}_{\sigma\in \Sigma   \cup \{\$\}    }, Q_a,Q_r,Q_g)$
    \Else
    \State return (not exist)$;$
    \EndIf
\end{algorithmic}
\end{algorithm}

Next we first demonstrate that the algorithm can find out a base $\cal{B}$ for Hilbert space $span\{v^*(x)| x\in\Sigma^*\}$.

\begin{proposition} \label{Algorithm2P}

In Algorithm 2 for learning MM-1QFA,  the final  set of vectors $\mathcal{B}$ consists of a base of Hilbert space $span\{v^*(x)| x\in\Sigma^*\}$, where $v^*(x)$ is actually the vector replied by oracle $\boldsymbol{AD}$ for input string $x$, that is $v^*(x)=\boldsymbol{AD}(x)=U(\sigma_k)P_nU(\sigma_{k-1})P_n\ldots U(\sigma_1)P_n|\psi_0\rangle$, for $x=\sigma_1\sigma_2\ldots\sigma_k\in\Sigma^*$.

\end{proposition}

\begin{proof}

Suppose that ${\cal B}=\{v^*(x_1), v^*(x_2),\ldots,v^*(x_m)\}$, where it is clear that some $x_i$ equals to $\varepsilon$. So, for any $x\in\Sigma^*$, there are  $x_j$ and $y\in\Sigma^*$ such that $x=x_jy$. It suffices to show that $v^*(x)$ can be linearly represented by the vectors in ${\cal B}$. By induction on the length $|y|$ of $y$. If $|y|=0$, i.e., $y=\varepsilon$, then it is obvious for $x=x_j$. In addition, for $|y|=1$, $v^*(x_jy)$ is linearly dependent on ${\cal B}$ in terms of the algorithm's operation.  Suppose that it holds for $|y|=k\geq 0$. Then we need to verify it holds for $|y|=k+1$. Denote $y=z\sigma$ with $|z|=k$. Then by induction hypothesis  $v^*(x_jz)= \sum_{k}c_kv^*(x_k)$ for some $c_k\in\mathbb{C}$ with $k=1,2,\ldots,m$. Therefore we have
\begin{align}
v^*(x)&=v^*(x_jz\sigma)\nonumber\\
&=U(\sigma)P_nv^*(x_jz)\nonumber\\
&=U(\sigma)P_n\sum_{k}c_kv^*(x_k)\nonumber\\
&=\sum_{k}c_kU(\sigma)P_nv^*(x_k)\nonumber\\
&=\sum_{k}c_kv^*(x_k\sigma).
 \end{align}

Since $v^*(x_k\sigma)$ is linearly dependent on ${\cal B}$ for $k=1,2,\ldots,m$, $v^*(x)$ can be linearly represented by the vectors in ${\cal B}$ and the proof is completed.

\end{proof}

Then we need to verify that the MM-1QFA ${\cal M}^*$ obtained in Algorithm 2 is  equivalent to the target MM-1QFA ${\cal M}$. This can be achieved by checking $V(\$)P_n|\psi_0^*\rangle=U(\$)P_n|\psi_0\rangle$ and  for any  $x=\sigma_1\sigma_2\ldots\sigma_k\in\Sigma^*$,
\begin{equation}\label{MM-1QFAequ}
V(\sigma_k)P_nV(\sigma_{k-1})P_n\ldots V(\sigma_1)P_n|\psi_0^*\rangle=U(\sigma_k)P_nU(\sigma_{k-1})P_n\ldots U(\sigma_1)P_n|\psi_0\rangle.
\end{equation}
 So we are going to prove the following theorem.

\begin{theorem}

In Algorithm 2 for learning MM-1QFA, we have
\begin{equation} \label{MM-1QFAequ2}
V(\$)P_n|\psi_0^*\rangle=U(\$)P_n|\psi_0\rangle;
\end{equation}
and for any $x=\sigma_1\sigma_2\ldots\sigma_k\in\Sigma^*$, Eq. (\ref{MM-1QFAequ}) holds, where $|x|\geq 1$.

\end{theorem}

\begin{proof}


Note $ v^*(\varepsilon) \in \mathcal{B}$, by means of Step 15 in Algorithm 2 and taking $\sigma=\$ $, we have
\begin{equation}
V(\$)P_nv^*(\varepsilon)=v^*(\$)=\boldsymbol{AD}(\$)=U(\$)P_n|\psi_0\rangle.
\end{equation}
Since $ v^*(\varepsilon)=\boldsymbol{AD}(\varepsilon)=|\psi_0\rangle$, and from Algorithm 2 we know $ \boldsymbol{AD}(\varepsilon)=|\psi_0^*\rangle$,  Eq. (\ref{MM-1QFAequ2}) holds.

Next we prove that Eq. (\ref{MM-1QFAequ}) holds for any $x=\sigma_1\sigma_2\ldots\sigma_k\in\Sigma^*$. We do it by induction method on the length of $|x|$.

If $|x|=1$, say $x=\sigma\in\Sigma$, then with Step 15 in Algorithm 2 and taking $v^*(\varepsilon)$, we have $V(\sigma)P_nv^*(\varepsilon)=v^*(\sigma)=\boldsymbol{AD}(\sigma)=U(\sigma)P_n|\psi_0\rangle$,  so, Eq. (\ref{MM-1QFAequ}) holds for $|x|=1$ due to $v^*(\varepsilon)=|\psi_0\rangle$.

Assume that  Eq. (\ref{MM-1QFAequ}) holds for any  $|x|=k\geq 1$. The rest is to prove that  Eq. (\ref{MM-1QFAequ}) holds for any  $|x|=k+1$. Let $x=y\sigma$ with $y=\sigma_1\sigma_2\ldots\sigma_k$. Suppose $v^*(y)=\sum_{i}c_iv^*(x_i)$ for some $c_k\in \mathbb{C}$. For each $i$,  by means of Step 15 in Algorithm 2, we have
\begin{equation}
V(\sigma)P_nv^*(x_i)=v^*(x_i\sigma)=\boldsymbol{AD}(x_i\sigma)=U(\sigma)P_n\boldsymbol{AD}(x_i),
\end{equation}
and therefore
\begin{equation}
V(\sigma)P_n\sum_{i}c_iv^*(x_i)=U(\sigma)P_n\sum_{i}c_i\boldsymbol{AD}(x_i).
\end{equation}
Since $v^*(x_i)=\boldsymbol{AD}(x_i)$, we further have
\begin{equation}
V(\sigma)P_nv^*(y)=U(\sigma)P_nv^*(y).
\end{equation}
By using $v^*(y)=U(\sigma_k)P_nU(\sigma_{k-1})P_n\ldots U(\sigma_1)P_n|\psi_0\rangle$, and the above induction hypothesis (i.e., Eq. (\ref{MM-1QFAequ}) holds), we have
\begin{equation}
V(\sigma_{k+1})P_nV(\sigma_k)P_nV(\sigma_{k-1})P_n\ldots V(\sigma_1)P_n|\psi_0^*\rangle=U(\sigma_{k+1})P_nU(\sigma_{k-1})P_n\ldots U(\sigma_1)P_n|\psi_0\rangle.
\end{equation}
Consequently, the proof is completed.




\end{proof}

To conclude the section,
 we give the computational complexity of Algorithm 2.

\begin{proposition} \label{Algorithm2:Complexity}

Let the target MM-1QFA to be learned have $n$'s bases states. The the computational complexity of Algorithm 2 is  $O(n^5|\Sigma|)$.

\end{proposition}

\begin{proof}
It is actually similar to the  proof of Proposition \ref{Algorithm1:Complexity}.
\end{proof}

\begin{remark}

Weighted finite automata (WFA) (e.g., see \cite{BM15})  are finite automata whose transitions and states are augmented with some weights, elements of a semiring, and a WFA also induces a function over strings. Learning WFA has been significantly studied and the details can be referred to \cite{BM15} and the references therein.  The algorithms for learning WFA are closely  related to  Hankel  matrices. More specifically, for a field $S$ and a finite alphabet  $\Sigma$, then the rank of the Hankel matrix $H_f$ associated to a function $f:\Sigma^*\rightarrow S$  is finite if and only if there exists a WFA ${\cal A }$ representing $f$ with rank($H_f$ ) states and no WFA representing $f$ admits fewer states.  Though  membership queries (MQ) and equivalence queries (EQ)  are used in in the algorithms of learning WFA, the way to induce functions  by WFA  is different from the definitions of accepting probabilities  in QFA, and particularly it is not known whether  the Hankel matrices can be used to study QFA (as we are aware,  there are no results  concerning the Hankel matrices associated with QFA). Therefore, it is still an open problem of  whether  the algorithms of learning WFA can be used to study learning QFA.

\end{remark}

\section{Concluding remarks}

Quantum finite automata (QFA) are simple models of quantum computing with finite memory, but QFA have significant advantages over classical finite automata concerning state complexity \cite{AY15,SY14,BK19},  and QFA can be realized physically to a considerable extent \cite{MP20}. As a new topic in quantum learning theory and quantum machine learning, learning QFA via queries has been studied in this paper.  As  classical model  learning \cite{Vaa17}, we can term it as {\it quantum model learning}.

The main results we have obtained are that we have proposed two polynomial-query  learning algorithms for measure-once one-way QFA (MO-1QFA) and measure-many one-way QFA (MM-1QFA), respectively. The oracle to be used is an $\boldsymbol{AD}$ oracle that can answer an amplitude distribution, and we have analyzed that a weaker oracle being only able to answer  accepting or rejecting for any inputting string may be not enough for learning QFA with polynomial time.

Here a question is how to compare $\boldsymbol{AD}$ oracle  with $\boldsymbol{MQ}$ oracle and $\boldsymbol{EQ}$ oracle? In general, $\boldsymbol{MQ}$ oracle and $\boldsymbol{EQ}$ oracle are together used in  classical models learning with deterministic or nondeterministic transformation of states; $\boldsymbol{AD}$ oracle  can return a superposition state for an input string in QFA, as $\boldsymbol{SD}$ oracle in \cite{Tze92} can return a distribution of state for an input string in PFA, so for learning DFA, $\boldsymbol{AD}$ oracle and $\boldsymbol{SD}$ oracle can return a state for each input string, not only accepting state as  $\boldsymbol{MQ}$ oracle can do.  Of course, the problem of whether both $\boldsymbol{MQ}$ oracle and $\boldsymbol{EQ}$ oracle  together can be used to study QFA learning is still not clear.  Furthermore,   if $\boldsymbol{AD}$ oracle can return the weight of general weighted automata for each input string,  then the problem of whether $\boldsymbol{AD}$ oracle can be used to study general weighted automata learning is worthy of consideration carefully.

However, we still do not know whether there is a  weaker oracle ${\cal Q}$  than $\boldsymbol{AD}$ oracle but by using ${\cal Q}$ one can learn MO-1QFA or MM-1QFA with polynomial time.  Furthermore, what is the weakest oracle to learn MO-1QFA and MM-1QFA with polynomial query complexity?
These are interesting and challenging problems to be solved. Of course,  for learning RFA, similar to learning DFA \cite{Ang87}, we can get an algorithm of  polynomial time by using  membership queries (MQ) together with equivalence queries (EQ).

Another interesting problem is how to realize these query oracles physically, including $\boldsymbol{SD}$ oracle, $\boldsymbol{AD}$ oracle, and even $\boldsymbol{MQ}$ as well as $\boldsymbol{EQ}$.  In quantum query algorithms, for any given Boolean function $f$,  it is supposed that a quantum query operator called an oracle, denoted by $O_f$, can output  the value of $f(x)$ with any input $x$.  As for the construction of quantum circuits for $O_f$, there are two cases:  (1)   If a Boolean function $f$ is in the form of disjunctive normal form (DNF) and  suppose that the truth table of $f$ is known, then an algorithm for constructing quantum circuit to realize $O_f$ was 
proposed in \cite{ACR21}. However, this method relies on the truth table of the function, which means that it is difficult to apply, since the truth table of the function is likely not known in practice.  (2)   If a Boolean function $f$ is a conjunctive normal form (CNF), then a polynomial-time algorithm was designed in \cite{QLX22} for constructing a quantum circuit to realize $O_f$, without any further condition on $f$.

As mentioned above,
besides MO-1QFA and MM-1QFA, there are other one-way QFA,  including
Latvian QFA  \cite{ABG06},  QFA with control language \cite{BMP03},    1QFA with ancilla qubits (1QFA-A) \cite{Pas00},   one-way quantum finite automata together with classical states (1QFAC) \cite{QLMS15}, and other 1QFA such as  Nayak-1QFA (Na-1QFA), General-1QFA (G-1QFA), and fully 1QFA (Ci-1QFA) etc.  So, one of the further problems worthy of consideration is to investigate learning these QFA via queries. 

Finally,  we would like to analyze partial possible methods for considering these problems. In the present paper, we have used   $\boldsymbol{AD}$ oracle as queries and quantum algorithms for determining the equivalence between 1QFA to be learned for learning  both MO-1QFA  and MM-1QFA. So, an algorithm for determining the equivalence between 1QFA to be learned is necessary in our method.  In general, as we studied in \cite{LQ08},  for designing an algorithm to determine the equivalence between 1QFA, we first transfer the 1QFA  to a classical linear mathematical model, and then  obtain the result by using the known  algorithm for determining the  equivalence   between classical linear mathematical models.
As pointed out in \cite{AY15},  the equivalence between 1QFA  can be determined, though the equivalence problems for some of 1QFA  still have not been studied carefully.  Of course, some 1QFA  also involve more parameters to be learned, for example, 1QFAC have classical states to be determined.

\section*{Acknowledgments}

The authors are grateful to the two anonymous referees for invaluable suggestions and comments that greatly helped us improve the quality of this paper.
This work is supported in part by the National Natural Science Foundation of China (Nos. 61876195, 61572532) and the Natural Science Foundation of Guangdong Province of China (No. 2017B030311011).


%

%

\section*{References}

\begin{thebibliography}{10}
\providecommand{\url}[1]{#1}
\csname url@samestyle\endcsname
\providecommand{\newblock}{\relax}
\providecommand{\bibinfo}[2]{#2}
\providecommand{\BIBentrySTDinterwordspacing}{\spaceskip=0pt\relax}
\providecommand{\BIBentryALTinterwordstretchfactor}{4}
\providecommand{\BIBentryALTinterwordspacing}{\spaceskip=\fontdimen2\font plus
\BIBentryALTinterwordstretchfactor\fontdimen3\font minus
  \fontdimen4\font\relax}
\providecommand{\BIBforeignlanguage}[2]{{%
\expandafter\ifx\csname l@#1\endcsname\relax
\typeout{** WARNING: IEEEtran.bst: No hyphenation pattern has been}%
\typeout{** loaded for the language `#1'. Using the pattern for}%
\typeout{** the default language instead.}%
\else
\language=\csname l@#1\endcsname
\fi
#2}}
\providecommand{\BIBdecl}{\relax}
\BIBdecl






























\bibitem{Aa07}
S. Aaronson, The learnability of quantum states, Proc. Math. Phys. Eng. Sci. 463 (2088) (2007) 3089–3114.


\bibitem{Ang82}
D. Angluin, Queries and concept learning, Journal of the ACM 29 (3) (1982) 741-765.


\bibitem{Ang87} D. Angluin, Learning regular sets from queries and counterexamples, Information and computation 75 (2) (1987) 87-106.

\bibitem{Ang88} D. Angluin, Queries and concept learning, Machine learning 2 (4) (1988) 319-342.



\bibitem{ABG06} A. Ambainis, M. Beaudry, M. Golovkins, A. Kikusts, M. Mercer, D. Therien, Algebraic results on quantum automata, Theory of Computing Systems 39 (1) (2006) 165–188.


\bibitem{ACR21}
J. Avron, O. Casper, and I. Rozen, Quantum advantage and noise reduction in distribute quantum computing, Physical Review A 104 (2021) 052404.



\bibitem{AEF15} D. Angluin, S. Eisenstat, D. Fisman, Learning regular languages via alternating automata, in: Proceeding of the 24th International Joint Conference on Artificial Intelligence, IJCAI 2015, AAAI Press, 2015, pp. 3308-3314.

\bibitem{AF98} A. Ambainis and R. Freivalds, One-way quantum finite automata: strengths, weaknesses and generalizations, in: Proceeding of the 39th Foundations of Computer Science, FOCS 98, IEEE Computer Society Press, 1998, pp. 332-341.


\bibitem{AW17} S. Arunachalam, R. de Wolf, A Survey of Quantum Learning Theory, ACM SIGACT News 48 (2) (2017) 41-67. Also, arXiv:1701.06806v3.


\bibitem{AW18}
S. Arunachalam, R. de Wolf, Optimal quantum sample complexity of learning
Algorithms, J Mach Learn Res 19 (2018) 1–36.








\bibitem{AY15}

A. Ambainis, A. Yakaryilmaz,  Automata and quantum computing,  in: Handbook of Automata Theory (II), edited by Jean-Eric Pin,   pp. 1457-1493,2021. arXiv:1507.01988.






\bibitem{Be15}
A. Belovs, Quantum algorithms for learning symmetric juntas via the adversary
Bound, Computational Complexity 24 (2) (2015) 255–293.



\bibitem{BHK09} B. Bollig, P. Habermehl, C. Kern, M. Leucker, Angluin-style learning of NFA, Proceeding of the 21st International Joint Conference on Artificial Intelligence, IJCAI-09, AAAI Press, 2009, pp. 1004-1009

\bibitem{BJ99} N.H. Bshouty, J.C. Jackson, Learning DNF over the uniform distribution using a quantum example oracle, SIAM Journal on Computing 28 (3) (1999) 1136–1153. Earlier version in COLT’95.

\bibitem{BK19} A.S. Bhatia, A. Kumar, Quantum finite automata: survey, status and research directions, arXiv:1901.07992v1, 2019.

\bibitem{BLL17} S. Berndt, M. Liśkiewicz, M. Lutter, R. Reischuk, Learning residual alternating automata, in: Proceeding of the 31st AAAI Conference on Artificial Intelligence, AAAI-17, AAAI Press, 2017, pp. 1749-1755.


\bibitem{Benedetti 2019} M. Benedetti, E. Lloyd, S. Sack,  M. Fiorentini,  Parameterized quantum circuits as machine learning models, Quantum Science and Technology 4 (4) (2019) 043001.
 


\bibitem{BM15} B.  Balle and M. Mohri, Learning Weighted Automata,     A. Maletti (Ed.): CAI 2015, LNCS 9270, pp. 1–21, 2015.

\bibitem{BMP03} A. Bertoni, C. Mereghetti, B. Palano, Quantum computing: 1-way quantum automata, in: Proceeding of the 7th International Conference on Developments in Language Theory, DLT 2003, Springer, 2003, pp. 1–20.

\bibitem{BP02} A. Brodsky, N. Pippenger, Characterizations of 1-way quantum finite automata, SIAM Journal on Computing 31 (5) (2002) 1456-1478.



\bibitem{BV96} Bergadano F, Varricchio S. Learning behaviors of automata from multiplicity and equivalence queries, SIAM Journal on Computing 25(6) (1996) 1268-1280.




\bibitem{BV04} S. Boyd, L. Vandenberghe, Convex Optimization, Cambridge University Press, Cambridge, 2004.


\bibitem{Biamonte  2017}J. Biamonte, P. Wittek, N. Pancotti, P. Rebentrost, N. Wiebe,  S. Lloyd,  Quantum machine learning, Nature 549 (7671) (2017) 195-202.



 \bibitem{CD16}I. Cong, L. Duan, Quantum discriminant analysis for dimensionality reduction and classification, New Journal of Physics 18 (7) (2016) 073011. 



\bibitem{CHY16}
H.C. Cheng, M.H. Hsieh, P.C. Yeh, The learnability of unknown quantum measurements, Quantum. Inf. Comput.16 (7\&8) (2016) 615–656.




\bibitem{FF63} D.K. Faddeev, V.N. Faddeev, Computational Methods of Linear Algebra, Freeman, San Francisco, 1963.

\bibitem{Ga12}
D. Gavinsky, Quantum predictive learning and communication complexity with single input, Quantum. Inf. Comput. 12 (7\&8) (2012) 575–588.


\bibitem{Gruska99} J. Gruska, Quantum Computing, McGraw-Hill, London, 1999.

\bibitem{Hig05} C. de la Higuera, A bibliographical study of grammatical inference, Pattern Recognition 38 (9) (2005) 1332–1348.

\bibitem{Hig10} C. de la Higuera, Grammatical Inference: Learning Automata and Grammars, Cambridge University Press, Cambridge, 2010.


\bibitem{HB22}H. Y. Huang, M. Broughton, J. Cotler, Quantum advantage in learning from experiments, Science 376 (6598) (2022) 1182–1186.


\bibitem{HH09}A. W. Harrow, A. Hassidim, S. Lloyd, Quantum algorithm for linear systems of equations, Physical Review Letters 103 (5) (2009) 150502.



\bibitem{Ka84} N. Karmarkar, A new polynomial-time algorithm for linear programming, Combinatorica 4 (1984) 373-395.

\bibitem{KP17}I. Kerenidis, A. Prakash, Quantum recommendation systems, Innovations in Theoretical Computer Science Conference 49 (2017) 1–21.


\bibitem{Krenn 2023}M. Krenn,  J. Landgraf,  T. Foesel,  F.  Marquardt, Artificial intelligence and machine learning for quantum technologies, Physical Review A 107 (1) (2023) 010101.

\bibitem{KV94} M.J. Kearns, U.V. Vazirani, An Introduction to Computational Learning Theory, MIT Press, 1994.

\bibitem{KW97} A. Kondacs, J. Watrous, On the power of quantum finite state automata, in: Proceedings of the 38th Annual IEEE Symposium on Foundations
of Computer Science, FOCS 97, IEEE Computer Society Press, 1997, pp. 66-75.


\bibitem{LR14}S. Lloyd, M. Mohseni, P. Rebentrost, Quantum principal component analysis, Nature Physics 10 (9) (2014) 108–113.



\bibitem{LQ06} L. Li, D.W. Qiu, Determination of equivalence between quantum sequential machines, Theoretical Computer Science 358 (2006) 65-74.

\bibitem{LQ08} L. Li, D.W. Qiu, Determining the equivalence for one-way quantum finite automata, Theoretical Computer Science  403 (2008) 42-51.




\bibitem{LW18}S. Lloyd, C. Weedbrook, Quantum generative adversarial learning,  Physical Review Letters 121 (4) (2018) 040502.


\bibitem{Marquardt 2021}F. Marquardt, Machine learning and quantum devices, SciPost Physics Lecture Notes 029 (2021).



\bibitem{Moore56} E. Moore, Gedanken-experiments on sequential machines, In Automata Studies, Annals of Mathematics Studies 34 (1956). Princeton University Press, 129-153.

\bibitem{MC00} C. Moore, J. P. Crutchfield, Quantum automata and quantum grammars, Theoretical Computer Science  237 (2000) 275-30.


\bibitem{Meyer 2023}J. J. Meyer, M. Mularski, E. Gil-Fuster,  A. A. Mele,  F. Arzani,  A. Wilms,  J.  Eisert, Exploiting symmetry in variational quantum machine learning, PRX Quantum 4 (1) (2023) 010328.


\bibitem{MM18}K. Mitarai, M. Negoro, M. Kitagawa, K. Fujii, Quantum circuit learning, Physical Review A 98 (3) (2018) 032309. 




\bibitem{MP20}
C. Mereghetti, B. Palano, S. Cialdi, et al., Photonic Realization of a Quantum Finite Automaton, PHYS REV RES, 2 (2020) Art. no. 013089.






\bibitem{MP20} C. Mereghetti, B. Palano, S. Cialdi, V. Vento, M.G.A. Paris, S. Olivares, Photonic Realization of a Quantum Finite Automaton, Physical Review Research 2  (2020) 013089 (15 pages).

\bibitem{MQL12} P. Mateus, D.W. Qiu, L. Li, On the complexity of minimizing probabilistic and quantum automata, Information and Computation 218 (2012) 36-53.



\bibitem{NC00} M. Nielsen, I. Chuang, Quantum Computation and Quantum Information, Cambridge University Press, Cambridge, 2000.

\bibitem{Pas00} K. Paschen, Quantum finite automata using ancilla qubits, University of Karlsruhe, Technical report, 2000.

\bibitem{Paz71} A. Paz, Introduction to Probabilistic Automata, Academic Press, Academic Press, New York, 1971.








\bibitem{PHYF2022}
S. Z. D. Plachta, M. Hiekkam{\"a}ki, A. Yakaryilmaz, R. Fickler, Quantum advantage using high-dimensional twisted photons as quantum finite automata,  QUANTUM-AUSTRIA, 6 (2022) 752.





\bibitem{Pel02} D. Peled, M. Vardi, M. Yannakakis, Black box checking, Journal of Automata Languages \& Combinatorics 7(2) (2002) 225-246.

\bibitem{Qiu07} D.W. Qiu, Automata theory based on quantum logic: Reversibilities and pushdown automata, Theoretical Computer Science  386 (1-2) (2007) 38-56.

\bibitem{QLMG12} D.W. Qiu, L. Li, P. Mateus, J. Gruska, Quantum Finite Automata, in: Finite State-Based Models and Applications (Edited by Jiacun Wang) , CRC Handbook. (2012) 113-144.

\bibitem{QLMS15} D.W. Qiu, L. Li, P. Mateus, and A. Sernadas, Exponentially more concise quantum recognition of non-RMM languages, Journal of Computer and System Sciences 81(2) (2015) 359-375.


\bibitem{QLX22} D.W. Qiu, L. Luo, L. Xiao,  Distributed Grover's algorithm, arXiv:2204.10487v4, 2022.


\bibitem{QLZMG11} D.W. Qiu, L. Li, X. Zou, P. Mateus, and J. Gruska, Multi-letter quantum finite automata: decidability of the equivalence and minimization of states, Acta Informatica 48 (2011) 271-290.

\bibitem{QY09} D.W. Qiu, S. Yu, Hierarchy and equivalence of multi-letter quantum finite automata, Theoretical Computer Science 410 (2009) 3006-3017.



\bibitem{Rebentrost 2014}P. Rebentrost, M. Mohseni,  S. Lloyd,  Quantum support vector machine for big data classification, Physical Review Letters 113 (13) (2014) 130503.


\bibitem{Rodriguez 2022} L. E. H. Rodriguez, A. Ullah,  K. J. R. Espinosa,  P. O. Dral,  A. A. Kananenka, A comparative study of different machine learning methods for dissipative quantum dynamics, Machine Learning: Science and Technology 3 (4) (2022) 045016.


\bibitem{Schuld  2015}M. Schuld, I. Sinayskiy,  F. Petruccione,  An introduction to quantum machine learning, Contemporary Physics 56 (2) (2015) 172–185.


\bibitem{SY14} A.C. Say, A. Yakaryılmaz, Quantum finite automata: A modern introduction, in: Computing with New Resources, Springer, 2014, pp. 208-222.


\bibitem{SG04}
R. Servedio, S. Gortler, Equivalences and separations between quantum and classical learnability, SIAM J. Comput. 33 (5) (2004) 1067–1092. 




\bibitem{Tze92} W.G. Tzeng, Learning probabilistic automata and Markov chains via queries, Machine Learning 8(2) (1992) 151-166.




\bibitem{TFLZZ2019} Y. Tian, T. Feng, M. Luo, S. Zheng, and X. Zhou,  Experimental demonstration of quantum finite automaton, NPJ QUANTUM INFORM,  5 (1) (2019) 1-5.





\bibitem{TK07}
 C. Tîrnăucă, T. Knuutila, Polynomial time algorithms for learning k-reversible languages and pattern languages with correction queries, in: M. Hutter,
R. Servedio, E. Takimoto (Eds.), Algorithmic Learning Theory: 18th International Conference, ALT’ 2007, in: Lecture Notes in Artificial Intelligence,
vol. 4754, Springer-Verlag, 2007, pp. 272–284.

\bibitem{Vaa17} F. Vaandrager, Model learning, Communications of the ACM  60 (2) (2017) 86–95.


\bibitem{WD12}N. Wiebe, D. Braun, S. Lloyd, Quantum algorithm for data fitting, Physical Review Letters 109 (5) (2012) 050505. 


\bibitem{WKS16}
N. Wiebe, A. Kapoor, K. Svore, Quantum deep learning, Quantum. Inf. Comput. 16 (7) (2016) 541–587. 


\bibitem{Zha10}
C. Zhang, An improved lower bound on query complexity for quantum PAC learning,
Inf. Process. Lett. 111 (1) (2010) 40–45. 



\bibitem{Zh19}Z. Zhao, A. Pozas-Kerstjens, P. Rebentrost, P. Wittek, Bayesian deep learning on a quantum computer, Quantum Machine Intelligence 1 (2019) 41–51.   








\end{thebibliography}

\end{document}